\documentclass[11point,sfbold]{paper}

\usepackage{geometry} 
\usepackage{graphicx} 
\usepackage{color}

\usepackage{booktabs} 
\usepackage{array} 
\usepackage{paralist} 
\usepackage{verbatim} 
\usepackage{subfigure} 
\usepackage{amsmath}
\usepackage{amssymb,amsthm}
\usepackage{ulem}

\usepackage{xspace}
\usepackage{mathrsfs}
\usepackage{amsopn}
\usepackage{tikz}

\definecolor{Red}{rgb}{1,0,0}
\definecolor{Blue}{rgb}{0,0,1}
\definecolor{Olive}{rgb}{0.41,0.55,0.13}
\definecolor{Green}{rgb}{0,1,0}
\definecolor{MGreen}{rgb}{0,0.8,0}
\definecolor{DGreen}{rgb}{0,0.55,0}
\definecolor{Yellow}{rgb}{1,1,0}
\definecolor{Cyan}{rgb}{0,1,1}
\definecolor{Magenta}{rgb}{1,0,1}
\definecolor{Orange}{rgb}{1,.5,0}
\definecolor{Violet}{rgb}{.5,0,.5}
\definecolor{Purple}{rgb}{.75,0,.25}
\definecolor{Brown}{rgb}{.75,.5,.25}
\definecolor{Grey}{rgb}{.5,.5,.5}

\theoremstyle{plain}

\newtheorem{theorem}{Theorem} 
\newtheorem{proposition}{Proposition}

\newtheorem{lemma}{Lemma}

\theoremstyle{remark}

\theoremstyle{definition}

\newcommand{\Rc}{\mathcal{R}}

\def\eps{\epsilon}
\def\la{\lambda}

\def\textiid{i.i.d.\@\xspace}
\newcommand\iid{\ifmmode\text{ i.i.d. } \else \textiid \fi}

\newcommand{\qmf}{\mathfrak{q}}

\newcommand{\beqs}{\begin{equation*}}
\newcommand{\eeqs}{\end{equation*}}
\newcommand{\beq}{\begin{equation}}
\newcommand{\eeq}{\end{equation}}
\usepackage{tikz}
\usepackage[hidelinks]{hyperref}
\usepackage{mathtools}
\usepackage{pgfplots}
\usepackage{cite}

\usepackage{etex}
\reserveinserts{28}

\begin{document}


\title{Sub-optimality of the Han--Kobayashi Achievable Region for Interference Channels}

\author{Chandra Nair \and Lingxiao Xia \and Mehdi Yazdanpanah
}
\institution{The Chinese University of Hong Kong}

\maketitle

\begin{abstract}
	Han--Kobayashi achievable region is the best known inner bound for a general discrete memoryless interference channel. We show that the capacity region can be strictly larger than the Han-Kobayashi region for some channel realizations, and hence the strict sub-optimality of Han--Kobayashi achievable region.
\end{abstract}

\allowdisplaybreaks

\section{Introduction}
	Interference channel models the communication of two (or more) transmitter/receiver pairs over a shared medium. A computable characterization of the capacity region ($\mathcal C$) for interference channels is a classical and a fundamental open problem in multi-terminal information theory. With the vast interest in wireless communications and the prominent presence of interference under such settings, characterizing (or understanding) the capacity region of the interference channel is very central.

	The interference channel shown below models the communication of two transmitter/receiver pairs and is the primary model we use for this study.
	
	\begin{figure}[h]\centering
		
		\begin{tikzpicture}[scale=1, every node/.style={scale=1}]
		\node at (0,1) {$M_1$};
		\node at (0,-1) {$M_2$};
		\draw [->,thick] (.4,1) -- (1.2,1);
		\draw [->,thick] (.4,-1) -- (1.2,-1);
		\draw (1.2,.5) rectangle +(2,1); \node at (2.2,1) {Encoder 1};
		\draw (1.2,.-1.5) rectangle +(2,1); \node at (2.2,-1) {Encoder 2};
		\draw [->,thick] (3.2,1)--(4,1); \node at (3.6,1.3) {$X_1^n$};
		\draw [->,thick] (3.2,-1)--(4,-1); \node at (3.6,-0.7) {$X_2^n$};
		\draw (4,-1.5) rectangle +(3.0,3); \node at (5.5,0) {$\qmf(y_1,y_2|x_1,x_2)$};
		\draw [->,thick] (7.0,1) --(7.8,1); \node at (7.4,1.3) {$Y_1^n$};
		\draw [->,thick] (7.0,-1) --(7.8,-1); \node at (7.4,-.7) {$Y_2^n$};
		\draw (7.8,.5) rectangle +(2,1); \node at (8.8,1) {Decoder 1};
		\draw (7.8,-1.5) rectangle +(2,1); \node at (8.8,-1) {Decoder 2};
		\draw [->,thick] (9.8,1) --(10.6,1); \node at (11.0,1) {$\hat M_1$};
		\draw [->,thick] (9.8,-1) --(10.6,-1); \node at (11.0,-1) {$\hat M_2$};
		\end{tikzpicture}
		\caption{Discrete memoryless interference channel}
		\label{fig:dmic}
	\end{figure}
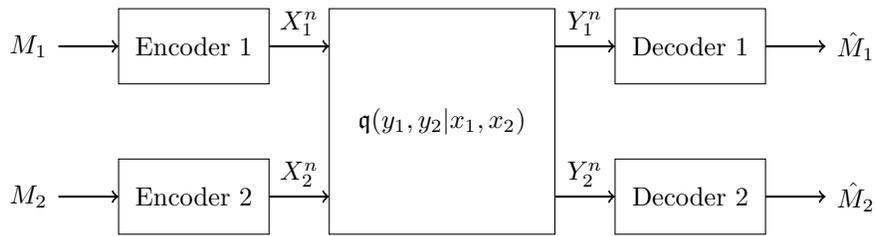
	We follow the standard definitions of achievable rates and capacity region, as well as the standard notations that may be found, for instance, in \cite{elk12}. The best known achievable rate region under such definitions is due to Han and Kobayashi \cite{hak81} and an equivalent form \cite{elk12} is presented below.

	\begin{theorem}[Han--Kobayashi (HK) inner bound]
		\label{th:hk}
		A rate-pair $(R_1,R_2)$ is achievable for the channel described in Figure \ref{fig:dmic} if
		\begin{align}
			R_1 &< I(X_1;Y_1|U_2,Q),\label{th:hkR1}\\
			R_2 &< I(X_2;Y_2|U_1,Q),\label{th:hkR2}\\
			R_1+R_2 &< I(X_1,U_2;Y_1|Q)+I(X_2;Y_2|U_1,U_2,Q),\label{th:hkR1R2}\\
			R_1+R_2 &< I(X_2,U_1;Y_2|Q)+I(X_1;Y_1|U_1,U_2,Q),\\
			R_1+R_2 &< I(X_1,U_2;Y_1|U_1,Q)+I(X_2,U_1;Y_2|U_2,Q),\\
			2R_1+R_2 &< I(X_1,U_2;Y_1|Q)+I(X_1;Y_1|U_1,U_2,Q)+I(X_2,U_1;Y_2|U_2,Q),\\
			R_1+2R_2 &< I(X_2,U_1;Y_2|Q)+I(X_2;Y_2`|U_1,U_2,Q)+I(X_1,U_2;Y_1|U_1,Q)
		\end{align}
		for some pmf $p(q)p(u_1,x_1|q)p(u_2,x_2|q)$, where $|U_1|\leq|X_1|+4$, $|U_2|\leq |X_2|+4$, and $|Q|\leq 7$.
		The set of achievable rate pairs form the Han--Kobayashi achievable region, or HK region, and is denoted by $\mathcal{R}_{hk}$. 
	\end{theorem}
	
	The capacity region is known under a small set of interference instantiations such as strong interference and injective  deterministic interference.  The sum capacity is established for a larger class of channels such as Gaussian interference channel with mixed or very weak interference. In all the cases mentioned above the capacity region (or the sum capacity) matches the one given by $\mathcal{R}_{hk}$. Furthermore, it was not known whether $\mathcal{R}_{hk}$ is the capacity region $\mathcal{C}$ or not. In this paper, we show that there are channel instances where $\mathcal{R}_{hk} \subsetneq \mathcal{C}$; thus showing the sub-optimality of the HK region.
	
	The main ingenuity of our work lies in the choice of the channel realizations because the computation of the HK region is not particularly straightforward. We study a class of interference channels, defined as CZI channels in the next section, where the evaluation of $\mathcal{R}_{hk}$ becomes significantly simplified\footnote{An earlier attempt was made by some of the authors \cite{lnx14c} along very similar lines where the sum-capacity of very weak interference channels was studied. However, they were unable to identify examples where the (normalized) two-letter achievable sum-rate of a two-letter product channel becomes larger than the original $\mathcal{R}_{hk}$. Note that CZI channels considered here are a further subclass of very weak interference channels; on the other hand we study a weighted sum-rate rather than the sum-rate.}. We take particular channels inside this class and compute a (normalized) two-letter achievable region of the corresponding two-letter product channel. We show that there are many examples where the (normalized) two-letter achievable region considered is strictly larger than $\mathcal{R}_{hk}$, which indicates $\mathcal{R}_{hk} \subsetneq \mathcal{C}$.
	
\section{{CZI Channel}}
	
	We say that an interference channel has clean Z interference (CZI) if one of the sub channels is a clean channel. We choose the channel from $X_2$ to $Y_2$ to be clean as depicted in Figure \ref{fig:zicoc} and study its HK region.
		\begin{figure}[h]\centering
		
			\begin{tikzpicture}[scale=1, every node/.style={scale=1}]
			\node at (0,1) {$M_1$};
			\node at (0,-1) {$M_2$};
			\draw [->,thick] (.4,1) -- (1.2,1);
			\draw [->,thick] (.4,-1) -- (1.2,-1);
			\draw (1.2,.5) rectangle +(2,1); \node at (2.2,1) {Encoder 1};
			\draw (1.2,.-1.5) rectangle +(2,1); \node at (2.2,-1) {Encoder 2};
			\draw [->,thick] (3.2,1)--(4.9,1); \node at (4,1.3) {$X_1^n$};
			\draw [->,thick] (3.2,-1)--(4.9,.9); \node at (4,-0.7) {$X_2^n$};
			\draw (4.9,.5) rectangle +(2.1,1); \node at (6,1) {$\qmf(y_1|x_1,x_2)$};
			\draw [->,thick] (7.0,1) --(7.8,1); \node at (7.4,1.3) {$Y_1^n$};
			\draw [->,thick] (3.2,-1) --(7.8,-1); \node at (6.8,-.7) {$Y_2^n=X_2^n$};
			\draw (7.8,.5) rectangle +(2,1); \node at (8.8,1) {Decoder 1};
			\draw (7.8,-1.5) rectangle +(2,1); \node at (8.8,-1) {Decoder 2};
			\draw [->,thick] (9.8,1) --(10.6,1); \node at (11.0,1) {$\hat M_1$};
			\draw [->,thick] (9.8,-1) --(10.6,-1); \node at (11.0,-1) {$\hat M_2$};
			\end{tikzpicture}
			\caption{ Discrete memoryless CZI channel}
			\label{fig:zicoc}
		\end{figure}
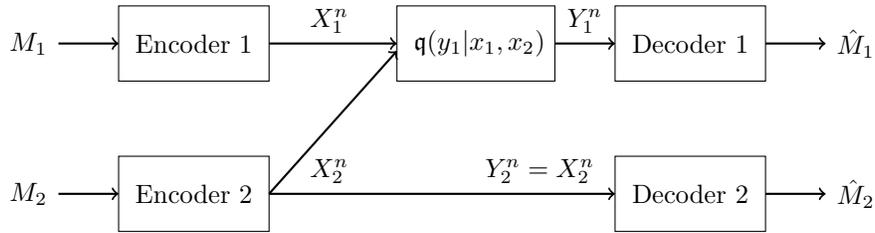
		
	The following proposition reveals an equivalent characterization of the HK region for CZI channels which simplifies its evaluation.
		
	\begin{proposition}
		\label{prop:hkz}
	    The HK region of a CZI channel is identical to the set of rate pairs $(R_1,R_2)$ that satisfy
		\begin{align}
			R_1 &< I(X_1;Y_1|U_2,Q),\label{1}\\
			R_2 &< H(X_2|Q),\label{2}\\
			R_1+R_2 &< I(X_1,U_2;Y_1|Q)+H(X_2|U_2,Q)\label{3}
		\end{align}
		for some pmf $p(q)p(u_2|q)p(x_2|u_2)p(x_1|q)$, where $|U_2|\leq |X_2|$ and $|Q|\leq 2$. 
	\end{proposition}
	
	\begin{proof}
	First of all, it is a simple exercise to note that the HK region of a CZI channel reduces to the three constraints above by setting $U_1 = \phi$. Hence, the above region is a subset of the HK region.
	
	Conversely, \eqref{1} is identical to \eqref{th:hkR1} of the HK region. \eqref{2} and \eqref{3} are respectively looser constraints than \eqref{th:hkR2} and \eqref{th:hkR1R2} of the HK region, which makes the above region larger than the original HK region. Thus proving equivalence.
	
	Note that the changes in cardinality of $U_2$ and $Q$ follow from standard applications of cardinality reduction techniques all while the underlying region remains the same. Therefore, we do not have to take these changes into account when talking about the two regions' equivalence. 
	\end{proof}
	
	The first result that we present below is a result that shows the {\it optimality} of the HK region along certain directions.
	
	\begin{proposition}
		\label{prop:2}
		For a CZI channel, 
		$$\underset{\mathcal{R}_{hk}}{\max}(\la R_1+R_2) = \underset{\mathcal C}{\max}(\la R_1+R_2), \forall \la \le 1.$$ 
	\end{proposition}
	\begin{proof}
	A standard converse/outer-bound argument proves that treating interference as noise is optimal.
	{\small\begin{align*}
	 n(\la R_1 + R_2)- n\epsilon &\overset{(a)}{\le}  H(X_2^n|X_1^n)+ \la I(X_1^n;Y_1^n) \\
	&= \sum_{i=1}^{n} H(X_{2i}|X_{1i}) - I(X_{2i};X_1^{n\smallsetminus i},X_2^{i-1}|X_{1i})+ \la I(X_1^n;Y_{1i}|Y_{1i+1}^n)\\
	&\le \sum_{i=1}^{n} H(X_{2i}|X_{1i}) - I(X_{2i};X_1^{n\smallsetminus i},X_2^{i-1}|X_{1i})+ \la I(X_1^n,Y_{1i+1}^n;Y_{1i})\\
	&= \sum_{i=1}^{n} H(X_{2i}|X_{1i}) - I(X_{2i};X_1^{n\smallsetminus i},X_2^{i-1}|X_{1i})\\
	&\qquad +\la\sum_{i=1}^{n} \left(I(X_1^n,Y_{1i+1}^n,X_2^{i-1};Y_{1i})- I(X_2^{i-1};Y_{1i}|X_1^n,Y_{1i+1}^n)\right)\\
	&\overset{(b)}{=} \sum_{i=1}^{n} H(X_{2i}|X_{1i}) - I(X_{2i};X_1^{n\smallsetminus i},X_2^{i-1}|X_{1i})\\
	&\qquad + \la\sum_{i=1}^{n} \left(I(X_1^n,Y_{1i+1}^n,X_2^{i-1};Y_{1i})- I(Y_{1i+1}^{n};X_{2i}|X_1^n,X_2^{i-1})\right)\\
	&= \sum_{i=1}^{n}  H(X_{2i}|X_{1i})- (1-\la)I(X_{2i};X_1^{n\smallsetminus i},X_2^{i-1}|X_{1i})\\
	&\qquad -\la\sum_{i=1}^{n} \Big(I(X_{2i};X_1^{n\smallsetminus i},X_2^{i-1},Y_{1i+1}^{n}|X_{1i})- I(X_1^n,Y_{1i+1}^n,X_2^{i-1};Y_{1i})\Big)\\
	&= \sum_{i=1}^{n} H(X_{2i}|X_{1i}) - (1-\la)I(X_{2i};X_1^{n\smallsetminus i},X_2^{i-1}|X_{1i})+ \la I(X_{1i};Y_{1i}) \\
	&\qquad -\la \sum_{i=1}^{n}\Big(I(X_{2i};X_1^{n\smallsetminus i},X_2^{i-1},Y_{1i+1}^n|X_{1i})- I(X_1^{n\smallsetminus i},X_2^{i-1},Y_{1i+1}^n;Y_{1i}|X_{1i})\Big)\\
	&\overset{(c)}{=} \sum_{i=1}^{n} H(X_{2i}|X_{1i})+ \la I(X_{1i};Y_{1i}) - (1-\la)I(X_{2i};X_1^{n\smallsetminus i},X_2^{i-1}|X_{1i})- \la I(X_{2i};X_1^{n\smallsetminus i},X_2^{i-1},Y_{1i+1}^n|Y_{1i},X_{1i})\\
	&\leq  n\left(\max ( H(X_2)+ \la I(X_1;Y_1)\right),
	\end{align*}}\\[-2em]
	
	\noindent where (a) follows from Fano's inequality, (b) Csiszar sum identity and (c) properties of the Markov chain formed by $Y_{1i}-\left(X_{1i},X_{2i}\right)-\left( X_2^{i-1},X_1^{n\smallsetminus i},Y_{1i+1}^n\right)$. 
	
	Since $\eps > 0$ is arbitrary, we see that any achievable rate pair must satisfy
	$$ \la R_1 + R_2 \leq \max_{p_1(x_1) p_2(x_2)}  H(X_2)+ \la I(X_1;Y_1),$$
	which is achievable by treating interference as noise, or more precisely, setting $U_2=\phi$ in the HK region. Hence, the proposition is established.
	\end{proof}

	On the contrary, we will see that, for some channels, 
	$$\underset{\mathcal{R}_{hk}}{\max}(\la R_1+R_2) < \underset{\mathcal C}{\max}(\la R_1+R_2)$$
	when $\la$ becomes larger than $1$. The following lemma helps us evaluate the quantity $\underset{\mathcal{R}_{hk}}{\max}(\la R_1+R_2)$.
	
	\begin{lemma}
		\label{lemma:2}
		For a CZI channel, for all $\la > 1$
		 \begin{align}
			&\underset{\mathcal{R}_{hk}}{\max}(\la R_1+R_2) = \underset{p_1(x_1)p_2(x_2)}{\max} \Big\{I(X_1,X_2;Y_1)  +\underset{p_2(x_2)}{\mathfrak{C}}\big[H(X_2)-I(X_2;Y_1|X_1) +(\lambda -1)I(X_1;Y_1)\big]\Big\}, \label{hkz}
		\end{align}
		 where $\underset{x}{\mathfrak{C}}[f(x)]$ of $f(x)$ denotes the upper concave envelope of $f(x)$ over $x$. \cite{nai13}
	\end{lemma}
	
	\begin{proof}	
	For any $(R_1,R_2) \in \Rc_{hk}$, there must exist a distribution $p(q) p_2(u_2,x_2|q) p_1(x_1|q) $ such that
	\begin{align*}
	\la R_1 + R_2 & \leq (\la-1) I(X_1;Y_1|U_2,Q)+ I(X_1,U_2;Y_1|Q) + H(X_2|U_2,Q) \\
	& = I(X_1,X_2;Y_2|Q) + H(X_2|U_2,Q)- I(X_2;Y_2|U_2,X_1,Q) + (\la-1) I(X_1;Y_1|U_2,Q)\\
	& \overset{(d)}{=} I(X_1,X_2;Y_2|Q)+\underset{p_2(x_2|q)}{\mathfrak{C}} \big[H(X_2|Q)-I(X_2;Y_1|X_1,Q) + (\lambda -1)I(X_1;Y_1|Q)\big],
	\end{align*}\\[-2em]
	
	\noindent where $(d)$ follows directly from the definition of the upper concave envelope.
	Since $Q$ computes an average, and since the average is less than the maximum, we obtain that
	\begin{align*}
	&\underset{\mathcal{R}_{hk}}{\max}(\la R_1+R_2)\leq \underset{p_1(x_1)p_2(x_2)}{\max} \Big\{I(X_1,X_2;Y_1) +\underset{p_2(x_2)}{\mathfrak{C}}\big[H(X_2)-I(X_2;Y_1|X_1) +(\lambda -1)I(X_1;Y_1)\big]\Big\}.
	\end{align*}
On the other hand, for any $p_2(u_2,x_2)p_1(x_1)$, the following rate pair
	$$ (R_1,R_2) = \left(I(X_1;Y_1|U_2), \; H(X_2|U_2) + I(U_2;Y_1)\right)$$
	belongs to $\Rc_{hk}$ as it satisfies the constraints. Thus,
	\begin{align*}
	\underset{\mathcal{R}_{hk}}{\max}(\la R_1+R_2)&\ge \max_{p_2(u_2,x_2)p_1(x_1) } \Big\{\la  I(X_1;Y_1|U_2)+ H(X_2|U_2)+I(U_2;Y_1)\Big\} \\
	& = \max_{p_2(u_2,x_2)p_1(x_1) }  \Big\{I(X_1,U_2;Y_1) + H(X_2|U_2)+(\la - 1)  I(X_1;Y_1|U_2)\Big\}\\
	& = \max_{p_2(u_2,x_2)p_1(x_1) } \Big\{I(X_1,X_2;Y_1) + H(X_2|U_2) - I(X_2;Y_1|U_2,X_1) +  (\la - 1)  I(X_1;Y_1|U_2)\Big\}\\
	&\overset{(e)}{=}\max_{p_2(x_2)p_1(x_1)}I(X_1,X_2;Y_1) + \underset{p_2(x_2)}{\mathfrak{C}} \big[H(X_2)-I(X_2;Y_1|X_1) + (\lambda -1)I(X_1;Y_1)\big],
	\end{align*}\\[-2em]

	\noindent where (e) also follows directly from the definition of the upper concave envelope, see \cite{nai13}. This establishes the converse and completes the proof of the lemma. 
	\end{proof}

	By viewing the channel use across two consecutive time-slots as the channel use of a single time-slot of the corresponding product channel, we obtain what is usually termed the {\it two-letter} realization of the original channel. For the two letter product channel of a CZI channel, the transition probability satisfies $$\tilde{\qmf}(y_{11}y_{12}|x_{11},x_{12}x_{21},x_{22})=\qmf(y_{11}|x_{11}x_{21})\qmf(y_{12}|x_{12},x_{22}),$$ where $\qmf$ is the transition probability of the CZI channel. 
	
	\begin{proposition}
		The set of rate pairs satisfying
		\begin{align*}
		R_1 &= \frac 12 I(X_{11},X_{12};Y_{11},Y_{12}|Q),\\
		R_2 &= \frac 12 H(X_{21},H_{22}|Q),
		\end{align*}
		for some pmf $p(q)p(x_{11},x_{12}|q)p(x_{21}x_{22}|q)$ with $|Q|\leq 2$ is achievable by the original channel. 
	\end{proposition}
	\begin{proof}
		This rate pair is precisely the {\it treating interference as noise} rate pair of the two-letter channel, and the normalization by $\frac 12$  is due to the fact that we code over two time-slots of the original channel.
	\end{proof}
	We denote this (normalized) two-letter HK region as $\mathcal{R}_{two}$.
	
	\subsection{Sub-optimality of the HK region}
	In this part we provide several CZI channels for which, for some fixed $(\la>1)$,  $\underset{\mathcal{R}_{two}}{\max}(\la R_1 +R_2)$ becomes larger than $\underset{\mathcal{R}_{hk}}{\max}(\la R_1 +R_2)$, which proves the sub-optimality of the HK region. 
	
	Examples are of channels with binary input and output. A $2\times2$ matrix is used to represent the channel:
	\begin{align*}
		\qmf(y_1|x_1,x_2) = \begin{bmatrix}
			P(Y_1=0|X_1,X_2=0,0)& P(Y_1=0|X_1,X_2=0,1)\\ 
			P(Y_1=0|X_1,X_2=1,0)& P(Y_1=0|X_1,X_2=1,1)
		\end{bmatrix}.
	\end{align*}
The fact that $X_2$ is binary allows us to compute the upper concave envelope in Lemma \ref{lemma:2} with extremely high precision.
	
	The channels in Table \ref{tab:tab1} are  obtained using numerical methods. We prove, as a demonstration, in the Appendix that the difference in rates of the first channel listed above is not due to numerical errors and that the maximum single-letter rate is indeed strictly smaller than the maximum (normalized) two-letter rate achieved by the corresponding two-letter product channel.
	
	\begin{table}[!]
	\caption{Table of counter-examples} \label{tab:tab1}
	\begin{center}
		\begin{tabular}{ | c | c | c | c | }
			\hline
			\rule{0pt}{3ex}\rule[-3ex]{0pt}{0pt} $\la$ & channel & $\underset{\mathcal{R}_{hk}}{\max}(\la R_1+R_2)$ & $\underset{\mathcal{R}_{two}}{\max}(\la R_1 +R_2)$ \\ \hline
		    \rule{0pt}{4ex}\rule[-3ex]{0pt}{0pt} 2 & $ \begin{bmatrix}
			1&                 0.5\\ 
			1&                  0 \end{bmatrix}$ &  1.107516 & 1.108141 \\ \hline 
			\rule{0pt}{4ex}\rule[-3ex]{0pt}{0pt} 2.5 & $ \begin{bmatrix}
			0.204581& 0.364813\\ 
			0.030209& 0.992978\end{bmatrix}$ & 1.159383 & 1.169312 \\ \hline
			\rule{0pt}{4ex}\rule[-3ex]{0pt}{0pt} 3 & $ \begin{bmatrix}
			0.591419& 0.865901\\ 
			0.004021& 0.898113\end{bmatrix}$ & 1.241521 & 1.255814 \\ \hline
			\rule{0pt}{4ex}\rule[-3ex]{0pt}{0pt} 3 & $ \begin{bmatrix}
			0.356166& 0.073253\\ 
			0.985504& 0.031707\end{bmatrix}$ & 1.292172 & 1.311027 \\ \hline 
			\rule{0pt}{4ex}\rule[-3ex]{0pt}{0pt} 3 & $ \begin{bmatrix}
			0.287272& 0.459966\\ 
			0.113711& 0.995405\end{bmatrix}$ & 1.117253 & 1.123151 \\ \hline
			\rule{0pt}{4ex}\rule[-3ex]{0pt}{0pt} 4 & $ \begin{bmatrix}
			0.429804& 0.147712\\ 
			0.948192& 0.002848\end{bmatrix}$ &  1.181392 & 1.196189 \\ \hline
			\rule{0pt}{4ex}\rule[-3ex]{0pt}{0pt} 4 & $ \begin{bmatrix}
			0.068730& 0.443630\\ 
			0.011377& 0.954887\end{bmatrix}$ & 1.223409 & 1.243958 \\ \hline
			\rule{0pt}{4ex}\rule[-3ex]{0pt}{0pt} 5 & $ \begin{bmatrix}
			0.969199& 0.564440\\ 
			0.954079& 0.061409\end{bmatrix}$ & 1.351229 & 1.372191 \\ \hline
			\rule{0pt}{4ex}\rule[-3ex]{0pt}{0pt} 5 & $ \begin{bmatrix}
			0.943226& 0.447252\\ 
			0.950791& 0.024302\end{bmatrix}$ & 1.231254 & 1.250564 \\ \hline 
			\rule{0pt}{4ex}\rule[-3ex]{0pt}{0pt} 6 & $ \begin{bmatrix}
			0.943292& 0.045996\\ 
			0.589551& 0.202487\end{bmatrix}$ & 1.069405 & 1.076932 \\ \hline 
			\rule{0pt}{4ex}\rule[-3ex]{0pt}{0pt} 6 & $ \begin{bmatrix}
			0.714431& 0.019375\\ 
			0.955918& 0.448539\end{bmatrix}$ & 1.528508 & 1.541781 \\ \hline
			\rule{0pt}{4ex}\rule[-3ex]{0pt}{0pt} 7 & $ \begin{bmatrix}
			0.058449& 0.558649\\ 
			0.194915& 0.959172\end{bmatrix}$ & 1.424974 & 1.452769 \\ \hline
			\rule{0pt}{4ex}\rule[-3ex]{0pt}{0pt} 7 & $ \begin{bmatrix}
			0.033312& 0.876067\\ 
			0.286125& 0.992825\end{bmatrix}$ &  1.179438 & 1.187867 \\ \hline
			\rule{0pt}{4ex}\rule[-3ex]{0pt}{0pt} 10 & $ \begin{bmatrix}
			0.307723& 0.874843\\ 
			0.032090& 0.710535\end{bmatrix}$ &   1.370830 & 1.388674 \\ \hline	
			\rule{0pt}{4ex}\rule[-3ex]{0pt}{0pt} 15 & $ \begin{bmatrix}
			0.946802& 0.311909\\ 
			0.730770& 0.155075\end{bmatrix}$ &  1.391596 & 1.406325 \\   \hline \rule{0pt}{4ex}\rule[-3ex]{0pt}{0pt} 100 & $ \begin{bmatrix}
			0.382410& 0.081474\\ 
			0.584797& 0.241840\end{bmatrix}$ & 3.754016 & 3.789316 \\ \hline
			\rule{0pt}{4ex}\rule[-3ex]{0pt}{0pt} 100 & $ \begin{bmatrix}
			0.673979& 0.194596\\ 
			0.781192& 0.285216\end{bmatrix}$ &   1.711938 & 1.730715 \\ \hline
		\end{tabular}
	\end{center}
	\end{table}

	\subsection{Intuition and a natural modification}\label{se:intuition}
	In this section, we present an intuition as well as a coding strategy motivated by this intuition that indicates how one may improve on the Han--Kobayashi encoding scheme.
	
	The counterexamples we generated in the last section had the following feature: even though $\lambda$ was strictly larger than one, the optimal $U_2$ that yielded $\underset{\mathcal{R}_{hk}}{\max}(\la R_1+R_2)$ was still the trivial random variable; implying that there were distributions $p_1(x_1)$ and $p_2(x_2)$ such that 
	$$ R_1 = I(X_1;Y_1), \quad R_2 = I(X_2;Y_2) = H(X_2)$$
	yielded the maximum weighted sum-rate.
	
	Suppose we now go to the two-letter product channel and take the product distribution of the marginals that yielded the one letter maximum as the transmitter distribution, clearly we would get the same rate. It is an easy exercise to verify that $I(X_1;Y_1)$ is convex in $X_2$ (utilizing the fact that $X_1$ and $X_2$ are independent). Thus a perturbation of the product distribution into two distributions that preserve the average would reduce $R_2 = \frac 12 H(X_{21},X_{22})$ but increase $R_1 = \frac 12 I(X_{11},X_{12};Y_{11},Y_{12}). $ Since we are interested in $\la R_1 + R_2$ with $\la > 1$, it is conceivable that such a perturbation would increase the weighted sum-rate. 
	
	Note that $X_2$ acts like a state variable on the communication of the channel between $X_1$ and $Y_1$. If the channel from $X_1 \to Y_1$, with $X_2$ as the state, is not memoryless, we know that the optimal code distributions on $X_1^n$ are not independent distributions. 
	
	For instance, if one creates $X_2^n$ according to a first order Markov process, the channel from $X_1^n$ to $Y_1^n$ becomes a channel whose state varies like a first order Markov process. For such a coding strategy, one could achieve
	$R_2 = \bar{H}(X_2)$, $R_1 = \bar{C}(X_1;Y_1)$, where $\bar{H}(X_2)$ denotes the entropy rate of the Markov process $X_2^n$ and $\bar{C}(X_1;Y_1)$ denotes the capacity of the channel whose state varies according to $X_2^n$.
	
	Note that in general $\bar{C}(X_1;Y_1)$ does not have a closed form and is quite hard to compute; but this scheme, as opposed to block coding, appears to the authors to be a natural fit for interference channels. It would also explain why i.i.d. coding (in the sense of Han--Kobayashi) might not be optimal for a CZI channel.

\section{Conclusion}
We have shown in the paper that Han--Kobayashi achievable region is strictly sub-optimal, which makes finding new ways of modeling achievable regions for interference channels almost a necessity in the future. 

	\section{Appendix}
	\subsection*{Analysis of a particular example}
	Consider the CZI channel where Figure \ref{fig:c1} the depicts $\qmf(y_1|x_1,x_2)$ as two point to point channels $X_1 \to Y_1$ for different choices of $X_2$. Our purpose is to show the details of computation $\Rc_{hk}$  when $\la = 2$.
	\begin{figure}[h]\centering
						\begin{tikzpicture}
						\node at (0,-0.65) {$X_2=0$};
						\node at (5,-0.65) {$X_2=1$};
						\node at (-1.7,0.75) {$X_1$};
						\node at (-1.3,1.5) {$0$};
						\node at (-1.3,0) {$1$};
						\filldraw
						(-1,1.5) circle (2pt)
						(-1,0) circle (2pt)
						(1.1,1.5) circle (2pt)
						(1.1,0) circle (2pt)
						(4,1.5) circle (2pt)
						(4,0) circle (2pt)
						(6.1,1.5) circle (2pt)
						(6.1,0) circle (2pt);
						\node at (1.7,0.75) {$Y_1$};
						\node at (1.4,1.5) {$0$};
						\node at (1.4,0) {$1$};
						\node at (3.3,0.75) {$X_1$};
						\node at (3.7,1.5) {$0$};
						\node at (3.7,0) {$1$};
						\node at (6.7,0.75) {$Y_1$};
						\node at (6.4,1.5) {$0$};
						\node at (6.4,0) {$1$};
						\draw [->,thick] (-1,1.5) -- (.9,1.5);
						\draw [->,thick] (-1,0) -- (1,1.4);
						\draw [->,thick] (4,1.5) -- (5.9,1.5) node[pos=0.5][align=center,  above]{$\frac 12$};
						\draw [->,thick] (4,0) -- (5.9,0);
						\draw [->,thick] (4,1.5) -- (6,.1) node[pos=0.4][align=right,  below]{$\frac 12$};
						\draw[dashed] (2.5,2) -- (2.5,-1);
						\end{tikzpicture}
						\caption{Binary CZI channel} %
						\label{fig:c1}
					\end{figure}
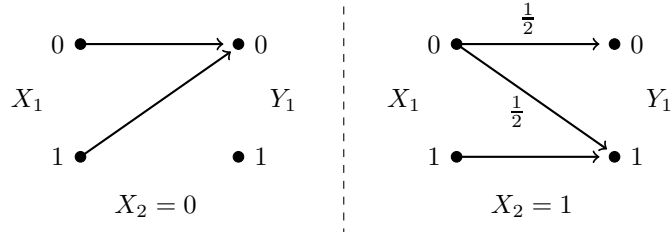

		 Let $P(X_1=0)=p$ and $P(X_2=0)=q$. By Lemma \ref{2}
		\begin{align}\label{hkla2}
		\underset{\mathcal{R}_{hk}}{\max}(2R_1+R_2) = \underset{p_1(x_1)p_2(x_2)}{\max} \Big\{I(X_1,X_2;Y_1)+\underset{p_2(x_2)}{\mathfrak{C}} \big[H(X_2)-I(X_2;Y_1|X_1) + I(X_1;Y_1)\big]\Big\}.
		\end{align}
		Define
	\begin{align}
	f(p,q) \coloneqq& H(X_2)-I(X_2;Y_1|X_1) + I(X_1;Y_1) \nonumber\\
	=& h_b(q)-2ph_b(\frac{q+1}{2})-2\bar{p}h_b(q) +h_b(q+\frac {p}{2}\bar{q})+p\bar{q} \label{f(p,q)}.
	\end{align}
	Here $h_b(x) = - x \log_2(x) - (1-x) \log_2(1-x)$ denotes the binary entropy function.
	Thus, we obtain that
	\begin{align}\label{hkla2a}
		\underset{\mathcal{R}_{hk}}{\max}(2R_1+R_2)\underset{p,q}{\max} ~ \Big \{h_b(q+\frac {p}{2}\bar{q}) - p\bar{q} + \underset{q}{\mathfrak{C}}  \big[h_b(q+\frac {p}{2}\bar{q})- p\bar{q}  \big]\Big\}.
		\end{align}
		Clearly the main computational imprecision may\footnote{In general since the concave envelope is computed over a single variable and the function is rather well behaved (at most two inflection points) when $X_2$ is binary, numerical computations using Matlab have yielded very high precision results even for the other counter examples listed.} arise from the estimation of the concave envelope; however as the next result shows; for this channel we obtain an explicit characterization of the concave envelope.

	\begin{lemma}
	\label{le:cenv}
		Consider the bivariate function $f(p,q)$ as defined in (\ref{f(p,q)}) where $(p,q) \in [0,1]\times [0,1]$.  Then
		\begin{itemize}
			\item[(i)] if $p > \frac 12$, $$\underset{q}{\mathfrak{C}}[f(p,q)]= f(p,q).$$
			\item[(ii)] if $p \le \frac 12$,
			{\small$$\hspace*{-0.3in}\underset{q}{\mathfrak{C}} [f(p,q)] = \begin{cases} 
			\begin{array}{cc}
			\frac{f(p,1-2p)-f(p,0)}{1-2p}q+f(p,0) & q \in [0,1-2p]\\\\
			f(p,q) &  1-2p < q 
			\end{array}
			\end{cases}.$$}
			
		\end{itemize}
	\end{lemma}
	\begin{proof}
	The second derivative  with respect to $q$ is
	\begin{align}
	\frac{\partial^2 f(p,q)}{\partial q^2} = \frac{p}{q \bar q \ln 2}   \frac{(1-3q - 2p\bar q)}{(1+q)(2q + p\bar q)} \label{secder}
	\end{align}
	If $p \in (\frac 12,1)$, then (\ref{secder}) is negative for $q \in (0,1)$, i.e., if $p > \frac 12$, $f(p,q)$ is concave in $q$ and $\underset{q}{\mathfrak{C}}[f(p,q)]= f(p,q)$.\\
	If $p \in (0,\frac 12)$, then (\ref{secder}) has one solution, $q^* \in (0,1)$. 
	$$q^* =  \frac{1-2p}{3-2p}.$$
	In fact, $f(p,q)$ is convex for $q \in (0,q*)$ and concave for $q \in (q*,1).$ Thus $\underset{q}{\mathfrak{C}} [f(p,q)]$ consists of two parts. First part is a tangent line from the point $f(p,0)$ to the function $f(p,\hat{q})$ and the second part is equal to $f(p,q)$.

	To find the point where the tangent line meets the function, ($\hat q$), we need to solve the following equation
	$$ \frac{f(p,\hat q)-f(p,0)}{\hat q} = \frac{\partial f(p,q)}{\partial q}\Big |_{\hat q}.$$
	Because the function is initially convex and then concave, the above equation will have at most one solution $\hat{q} \neq 0$. 
	One can verify that $\hat q = 1-2p$ is the required solution, and this completes the proof. \end{proof}
	
	Define $F(p,q)$ for $(p,q) \in [0,1] \times [0,1]$ as
	\begin{equation}
	\begin{cases} \begin{array}{lc} h_b(q+\frac {p}{2}\bar{q}) - p\bar{q} + f(p,q) & q \geq \min\{0,1-2p\}\\
	 h_b(q+\frac {p}{2}\bar{q}) - p\bar{q} & \\
	 \quad + \frac{f(p,1-2p)-f(p,0)}{1-2p}q+f(p,0) & o.w.,
	 \end{array}
	 \end{cases}
	\label{eq:actexp}
	\end{equation}
	where $f(p,q)$ is defined in \eqref{f(p,q)}.
	
	From Lemma \ref{le:cenv} and  \eqref{hkla2a}, we know that 
	\begin{align}\label{hkla3}
		\underset{\mathcal{R}_{hk}}{\max}(2R_1+R_2) = \underset{p,q}{\max} ~F(p,q).
		\end{align}
		A tedious exercise shows that the concave envelope of $F(p,q)$ w.r.t. $(p,q)$ matches the function value $F(p_0,q_0)$ at\footnote{Naturally, we choose a point that is numerically very close to the true maximum. } $(p_0,q_0) = (0.507829413, 0.436538150)$.  Hence an upper bound on $\underset{\mathcal{R}_{hk}}{\max}(2R_1+R_2)$ is given by maximum value of the supporting hyperplane to $F(p,q)$ at $p_0,q_0$, which is in turn upper bounded by $F(p_0,q_0) + |a| + |b|$ where $a = \frac{\partial F}{\partial p}\Big|_{p_0}$, and $b=  \frac{\partial F}{\partial q}\Big|_{q_0}. $ Evaluating the values we obtain an upper bound given by
		\begin{equation} \underset{\mathcal{R}_{hk}}{\max}(2R_1+R_2) \leq   1.107577. \label{eq:bd1let} \end{equation}
		On the other hand consider the following point in $\Rc_{two}$ given by
		$$ R_1 = \frac 12I(X_{11},X_{12};Y_{11},Y_{12}),
			\quad R_2 = \frac 12H(X_{21},H_{22}|Q),    $$
			where $P((X_{11},X_{12}) = (0,0)) = p_0, P((X_{11},X_{12}) = (1,1)) = 1-p_0$, and 
			$ P((X_{21},X_{22})=(0,0) ) = 0.36 q_0, P((X_{21},X_{22})=(0,1) ) = P((X_{21},X_{22})=(1,0) )  = 0.64 q_0, P((X_{21},X_{22})=(1,1) )  = 1-1.64 q_0$. For this choice of distribution we get $2 R_1 + R_2 = 1.1080356$, which is strictly larger than the bound given in \eqref{eq:bd1let}. This establishes the sub-optimality of the Han--Kobayashi region for the particular example considered in the Appendix.
			
			As mentioned in Section \ref{se:intuition} the distribution of $(X_{21},X_{22})$ that outperforms the one-letter region is not the product distribution; but more surprisingly one is doing {\it repetition coding} on $X_{11}, X_{12}$.
				 
\bibliographystyle{amsplain}
\bibliography{mybiblio}

\end{document}